\documentclass[12pt,journal,onecolumn]{IEEEtran}
\usepackage{amsmath,amsfonts}
\usepackage{algorithmic}
\usepackage{algorithm}
\usepackage{array}
\usepackage[caption=false,font=normalsize,labelfont=sf,textfont=sf]{subfig}
\usepackage{textcomp}
\usepackage{stfloats}
\usepackage{url}
\usepackage{verbatim}
\usepackage{graphicx}
\usepackage{cite}
\usepackage{multirow}
\usepackage{tabularx}
\usepackage{booktabs}

\usepackage{amssymb}
\usepackage{amsthm}
\usepackage{booktabs}
\usepackage{enumitem}
\usepackage{xcolor}
\usepackage{hyperref}
\hypersetup{
  hypertexnames=false,
  colorlinks=true,
  allcolors=blue
}

\theoremstyle{plain}
\newtheorem{theorem}{Theorem}[section]
\newtheorem{proposition}[theorem]{Proposition}

\newtheorem{corollary}[theorem]{Corollary}

\theoremstyle{definition}

\newtheorem{definition}{Definition}

\theoremstyle{remark}

\DeclareMathOperator{\XOR}{\oplus}
\DeclareMathOperator{\SBIBD}{SBIBD}
\newcommand{\bbF}{\mathbb{F}}

\begin{document}

\title{Replica Thresholds for Stripeless Erasure Coding Based on Symmetric Block Designs}

\author{Guiwen Luo, Xiao Li, Yaqiang Zhang, and Ke Wang%
\thanks{The authors are with
IEIT SYSTEMS Co., Ltd., Jinan 250101, China.
Email: \{luoguiwen,\allowbreak lixiao09,\allowbreak
zhangyaqiang,\allowbreak wangke12\}@ieisystem.com.
Corresponding author: Xiao Li.}%
}

\markboth{IEEE Transactions on Information Theory}%
{Replica Thresholds for Stripeless Erasure Coding Based on Symmetric Block Designs}


\maketitle

\begin{abstract}
This paper investigates a fundamental question in stripeless erasure coding based on symmetric balanced incomplete block designs (SBIBDs): how many replicas per object are precisely required to guarantee recovery from any set of at most $p$ node failures? We refine the known
sufficient recovery guarantee for the generalized $\SBIBD(v,k,\lambda)$ construction. One replica is necessary and sufficient for $p=1$, and 
$q=\lambda(p-1)+2$ replicas ($q\le k$) guarantee recovery for $p\ge2$. Recovery takes one round for  $p=2$ and at most two
rounds in general. To study the tightness of this count, we define the universal replica threshold $q^\ast(A,p)$ for a fixed zero-diagonal SBIBD
representative~$A$. We determine $q^\ast(A,p)$  for $p=1$ and  $p=2$. For $p\ge3$, we identify pairwise separated sets and private target
permutations as the structures determining whether the sufficient
count is tight or can be further reduced. These conditions give exact
thresholds for all but the case where $\lambda>1$ and $A$ contains no
pairwise separated set of size~$p$. We improve the bounds for this remaining
case and leave its exact threshold open. Finally, we give sufficient conditions
for pairwise separated sets and show how affinity relabeling can realize private target permutations.
\end{abstract}


\section{Introduction}\label{sec:intro}
\IEEEPARstart{C}onventional erasure coding organizes data and parity chunks into fixed
stripes. Each data chunk has a predetermined set of companion chunks, and
recovery is performed within this fixed group. Stripeless erasure coding
removes this persistent stripe membership. Each object retains one primary
copy and sends temporary replicas to selected target nodes. A target node
combines replicas received from different primary nodes into an XOR parity
and then discards the constituent replicas. Consequently, objects can be
placed and encoded independently without first being assigned to fixed
stripes.

Removing stripe boundaries is particularly useful for distributed storage
systems that keep many small objects in memory. In such
systems, splitting objects or collecting multiple objects into stripes,
and maintaining stripe metadata can introduce coordination and data placement
overheads that are significant relative to memory and network latency.
Stripeless coding replaces global stripe construction with replication at the object level and local parity formation at each node. Its main target is object storage with low latency in which the management cost of fixed stripes is a
major concern. Conventional stripes may remain preferable for large objects
and workloads optimized for bandwidth that benefit from mature MDS coding and
sequential data access.


Gao et al.~\cite{nos2025} first introduced Nos, a stripeless data placement scheme
for distributed storage systems that keep small objects in memory. Let $p$ denote the maximum number of simultaneous node failures to be
tolerated. For each
object primary on node~$i$, Nos independently selects $p+1$ target nodes from
the $k$ nodes permitted by a zero-diagonal $\SBIBD(v,k,1)$ representative.
Each target node forms a parity by XORing one replica from each of its $k$
prescribed source nodes. This organization removes persistent stripe
membership and allows replica placement and parity formation without a
centralized metadata service. The intersection property with $\lambda=1$
supports recovery from any set of at most $p\ (p<k)$ failed nodes. However,  Nos leaves two issues unresolved. First, the flexibility of Nos is limited by its design parameters. The
condition $\lambda=1$ forces $v=k(k-1)+1$, so the number of design nodes grows quadratically with~$k$ and only certain combinations of cluster size and parity width are available. Second, Nos uses
$q=p+1$ as a uniform sufficient replica count but does not determine whether this count is tight.

FLESC~\cite{flesc2026} alleviates the parameter limitation of Nos in two
complementary ways. At the scheme level, its coding scheme replaces
$\SBIBD(v,k,1)$ by $\SBIBD(v,k,\lambda)$ with a general admissible value of
$\lambda$. It proves that
$q=\lambda(p-1)+2$
replicas per object suffice when $k\ge\lambda(p-1)+2$. When $\lambda=1$,
this count becomes $p+1$ and the scheme reduces to Nos. Allowing larger
values of~$\lambda$ changes the design size to
$v={k(k-1)}/{\lambda}+1$,
which expands the available parameter choices when the corresponding SBIBDs exist. At the deployment level, FLESC introduces a virtual node layer that separates the design nodes from the physical storage nodes. Nevertheless, the threshold
question remains unresolved.  Like Nos, FLESC does not investigate whether its sufficient replica count is
tight or can be reduced for a fixed SBIBD representative.

This paper focuses on determining the exact universal replica threshold
for recovery from any set of at most $p$ node failures. For a fixed zero-diagonal SBIBD representative~$A$, we define the
universal replica threshold $q^*(A,p)$ as the minimum replica count that
guarantees recovery for every legal placement and every failure set of at most size~$p$. Compared to FLESC, we sharpen its sufficient recovery guarantee by treating $p=1$, $p=2$ and $p\ge 3$ separately and by making the recovery rounds explicit. We
then determine when this sufficient replica count is tight and when it can be reduced. The resulting classification shows that the threshold depends not only on the 
$\SBIBD(v, k, \lambda)$, but also on structural properties of the chosen representative matrix.

Our main contributions in this paper are as follows.

\medskip\noindent\textbf{Contributions.}
\begin{itemize}[nosep]
  \item We refine the sufficient recovery guarantee established
    in FLESC. For $p=1$, we prove that one replica per object is
    necessary and sufficient. For $p=2$, we show that every lost object
    is directly recoverable under the sufficient replica count. For $p\ge3$, we make the two round recovery structure explicit and give
    a more rigorous and complete proof that applies to every legal placement.

  \item For a fixed zero-diagonal SBIBD representative~$A$, we define and
    study the universal replica threshold $q^*(A,p)$. For $p=1$, we determine the exact threshold $q^*(A,1)=1$. For $p=2$, the exact threshold is classified by introducing \emph{reciprocal pairs}. For $p\ge 3$, we derive general bounds and classify the threshold using \emph{pairwise separated sets} and \emph{private target permutations}.
    Table~\ref{tab:threshold-summary} summarizes the resulting
    classification, which is exact except when $p\ge3$, $\lambda>1$, and $A$
    contains no pairwise separated set of size~$p$.

  \item We investigate when the structures used in the threshold
    classification can occur. We give a sufficient condition and a greedy
    construction for pairwise separated sets, establish their existence for
    sufficiently large feasible values of~$k$, and show how suitable affinity
    relabelings can realize reciprocal pairs and private target permutations.
\end{itemize}

We further discuss the design implications of these results. General admissible values of~$\lambda$ provide greater flexibility in the cluster size. We also quantify the supported failure tolerance, encoding traffic, and persistent parity storage, and show that selecting a suitable representative reduces the replica count by one compared to Nos and FLESC.
%
%

\medskip\noindent\textbf{Other uses of block designs in coding and storage.}
Block designs have long been used for different purposes in coding and storage.
Smith~\cite{smith1968} used incomplete block designs to construct codes with
majority logic decoding, while Kou et al.~\cite{kou2001} used incidence
structures from finite geometries to construct LDPC codes. In distributed
storage, Tian et al.~\cite{tian2013,tian2015} used block designs to organize
shorter erasure codes for exact node repair. Fractional repetition codes also
use Steiner systems and other combinatorial designs to place repeated coded packets and support uncoded exact repair~\cite{elrouayheb2010,olmez2016}.
Block designs have further been used to distribute parity groups across disk
arrays and balance the reconstruction load after a disk failure
\cite{holland1992}. More recently, Rukavina and Tonchev~\cite{rukavina2024} classified a family
of symmetric designs and studied the ternary linear codes generated by their
incidence matrices.

These studies use block designs to define parity constraints, distribute coded
packets, organize repair groups, or balance disk layouts. Our work uses
an SBIBD representative differently. It specifies the admissible relationships
between primary nodes and target nodes, while each object may independently
select its target set. 

\section{Preliminaries}\label{sec:prelim}

\subsection{Symmetric balanced incomplete block designs}\label{sec:sbibd}

\begin{definition}[Symmetric balanced incomplete block design]\label{def:sbibd}
For integers $1\le\lambda<k<v$, an $\SBIBD(v,k,\lambda)$ consists of $v$
points and $v$ blocks such that every block contains
$k$ points, every point lies in $k$ blocks, and every two distinct points lie
together in exactly $\lambda$ blocks. After independently labeling the points
and blocks by $0,1,\cdots,v-1$, its incidence matrix
$A=(A_{ij})\in\{0,1\}^{v\times v}$ is defined by $A_{ij}=1$ if and only if
point~$i$ lies in block~$j$. Such an independent labeling is called an \emph{affinity labeling} and the
resulting incidence matrix a \emph{representative} of the SBIBD. An affinity
labeling is said to have a \emph{zero diagonal} if its representative satisfies
$A_{ii}=0$ for every $i\in\{0,1,\cdots,v-1\}$.
\end{definition}

Every $\SBIBD(v,k,\lambda)$ admits a zero-diagonal representative.
A proof is given in Appendix~\ref{app:zero-diagonal}.
Unless stated otherwise, the representative matrices considered in this paper are zero-diagonal. Counting pairs consisting of an unordered point pair and a block containing
that pair gives $\lambda\binom{v}{2}=v\binom{k}{2}$, and hence
\begin{equation}\label{eq:sbibd-identity}
  \lambda(v-1)=k(k-1),\ v=\frac{k(k-1)}{\lambda}+1.
\end{equation}

The identity \eqref{eq:sbibd-identity} is necessary but not sufficient for
the existence of an $\SBIBD(v,k,\lambda)$. General existence results and
constructions are surveyed in~\cite{lindner2017,colbourn2010}.

\subsection{Stripeless erasure coding on an SBIBD}\label{sec:stripeless}

Suppose we have a storage cluster consisting of $v$ nodes indexed by $0,1,\cdots,v-1$. We utilize the block design to organize replica placement and parity formation in the coding scheme. Each node plays both roles in this scheme. It
may store primary copies of data objects and also host parities formed from
replicas received from other nodes. In the representative~$A$, row~$i$
represents node~$i$ as a primary node, while column~$j$
represents node~$j$ as a parity host. $A_{ij}=1$ means
that an object primary on node~$i$ is allowed to send a replica to node~$j$
for parity formation.

For a node~$j$, define its replication source set by
\begin{equation*}
\mathcal{S}_{\mathrm{src}}(j):=\{i:A_{ij}=1\},
\end{equation*}
and for a node~$i$, define its replication target set by
\begin{equation*}
\mathcal{S}_{\mathrm{tgt}}(i):=\{j:A_{ij}=1\}.
\end{equation*}
The first set lists the primary nodes from which node~$j$ may receive
replicas. The second lists the target nodes available to objects primary
on node~$i$. Both sets have size~$k$. Since $A$ has zero diagonal,
\[
i\notin\mathcal{S}_{\mathrm{tgt}}(i).
\]
For distinct nodes $i$ and~$j$, let
\begin{equation}\label{eq:Bij}
B_{ij}:=
\mathcal{S}_{\mathrm{tgt}}(i)
\cap
\mathcal{S}_{\mathrm{tgt}}(j).
\end{equation}
\(B_{ij}\) is the set of common admissible target nodes for objects primary on node \(i\) and objects primary on node \(j\). By the SBIBD property, 
\[|B_{ij}|=\lambda,\ B_{ij}=B_{ji},\ B_{ij}\cap\{i,j\}=\varnothing.\]

We assume all data objects are binary words of equal length, and XOR denotes componentwise
addition over $\bbF_2$. Each object \(x\) has exactly one primary copy, stored at its primary node \(i\).
For an integer $q\in\{0,1,\cdots,k\}$, the object selects a set
$T_x\subseteq\mathcal{S}_{\mathrm{tgt}}(i)$ of $q$ distinct target nodes and
sends one replica to each node in~$T_x$.  The set \(\mathcal S_{\mathrm{tgt}}(i)\) is common to all objects primary on node \(i\), whereas \(T_x\) is selected separately for each object and may differ between such objects. Thus \(\mathcal S_{\mathrm{tgt}}(i)\) contains all \(k\) admissible targets, while \(T_x\) contains the \(q\) targets actually used by \(x\). A smaller \(q\) reduces replication traffic, while a larger \(q\)  provides greater recovery redundancy.

Incoming replicas at node~$j$ are temporarily buffered according to their
primary nodes. A placement is \emph{legal} if at each node~$j$ the replicas received
at~$j$ are partitioned into $k$-tuples, each containing exactly one replica
from every primary node in $\mathcal{S}_{\mathrm{src}}(j)$. We analyze only
completed legal placements. An implementation may also buffer unmatched
replicas or store partial parities.

If one such tuple contains the objects
$x_1,x_2,\cdots,x_k$, node~$j$ stores
\[
P=x_1\XOR x_2\XOR\cdots\XOR x_k
\]
as one parity. The parity \emph{encodes} these objects, which are called its
\emph{encodees}. After the parity has been formed, the replicas used to form
it are discarded. Each object $x$ occurs in
exactly one parity at every node in its selected set~$T_x$.
Thus $q$ is also the number of parities encoding~$x$.

\begin{definition}[Node, object, and parity states]
Fix a set \(F\subseteq\{0,1,\cdots,v-1\}\) of failed nodes. Nodes in~$F$ are \emph{failed}, and all other nodes are
\emph{alive}. An object is \emph{lost} if its primary node is failed, and a
parity is \emph{lost} if the node storing it is failed. Here ``lost'' means
unavailable at its original location, not necessarily unrecoverable. An object
or parity that is not lost is called \emph{alive}.
\end{definition}

Given a set of failed nodes, the surviving primary copies and parities form a system of XOR equations. The lost objects are \emph{recoverable} if their values are uniquely determined by this surviving information.
\begin{definition}[Direct recovery]\label{def:direct-recovery}
A lost object~$x$ is \emph{directly recoverable} if some alive parity~$P$
encoding~$x$ has no other lost encodee. If
$P=x\XOR x_1\XOR x_2\XOR\cdots\XOR x_{k-1}$, then
$x=P\XOR x_1\XOR x_2\XOR\cdots\XOR x_{k-1}$, where every $x_\ell$ can be fetched
from its alive primary node.
\end{definition}

Recovery is organized into rounds. The first round may use only alive primary
copies and alive parities. Later rounds may also use values recovered in
earlier rounds. Thus direct recovery is recovery in the first round.

\section{Universal recovery guarantee}\label{sec:main}
FLESC~\cite{flesc2026} established the sufficient count
$\lambda(p-1)+2$ for multiple node failures through a counting argument.
We refine that result by treating one and two failures separately, making the
recovery rounds explicit, and formulating the guarantee uniformly over all
legal placements. 

For $p\ge1$, define
\begin{equation}\label{eq:q-piecewise}
q=
\begin{cases}
1, & p=1,\\
\lambda(p-1)+2, & p\ge2.
\end{cases}
\end{equation}

\begin{theorem}[Universal sufficient replica count]\label{thm:main}
Let $p\ge1$ be an integer, let $A$ be a fixed zero-diagonal representative of an
$\SBIBD(v,k,\lambda)$, and let $q$ be given by~\eqref{eq:q-piecewise}. If
$q\le k$, then every legal placement using exactly $q$ distinct target nodes
per object permits recovery after the failure of any set of at most $p$ nodes.
If at most two nodes fail, every lost object is directly recoverable. In
general, all directly recoverable lost objects can be restored in the first
round and every remaining lost object in the second round.
\end{theorem}

\begin{proof}
Let $F$ be a failure set with $|F|\le p$. The case $F=\varnothing$ is
trivial.

Suppose $F=\{i\}$. An object $x$ whose primary node is alive can be read
directly. If the primary node of $x$ is $i$, then
$T_x\subseteq\mathcal S_{\mathrm{tgt}}(i)$ and
$i\notin\mathcal S_{\mathrm{tgt}}(i)$ because $A_{ii}=0$. Hence every parity encoding~$x$ is alive and has no other lost encodee. Thus $x$ is directly recoverable.

Now let $F=\{i,j\}$, and let $x$ be primary on node~$i$. Since $i\notin T_x$, at most one parity encoding~$x$ is lost, namely the one hosted
on node~$j$ if $j\in T_x$. Any parity encoding~$x$ that contains another lost
encodee must be hosted in~$B_{ij}$. Hence at most
$1+|B_{ij}|=\lambda+1$ parities are unusable for direct recovery.
Since $p\ge2$ implies
$q=\lambda(p-1)+2\ge\lambda+2$, at least one parity of~$x$ is alive and
contains no other lost encodee. Thus $x$ is directly recoverable.

It remains to consider $|F|\ge3$. If $ |F|<p$, enlarge $F$ to a set of size~$p$. Recovery under the enlarged failure set also gives recovery under the original failure set. We may therefore assume that $|F|=p$. Without loss of generality, assume that node~$0$ is failed, and let $x$ be an object primary on node~$0$.

Object~$x$ has $q = \lambda(p-1)+2$ parities, one on each of the
$q$ distinct selected target nodes $T_x$. Let
\[
L_x:=F\cap T_x,\ \ell_x:=|L_x|.
\]
Thus $L_x$ is the set of failed hosts of parities encoding~$x$, and
$0\le\ell_x\le p-1$ because $A_{00}=0$. The remaining $q-\ell_x$ parities encoding $x$ are alive.

If some alive parity~$P_0$ of~$x$ has only directly recoverable lost encodees
other than~$x$, recover those encodees in the first round and then recover~$x$. Suppose, for a contradiction, that no such $P_0$ exists. From every alive
parity of~$x$, choose a lost encodee $y\ne x$ that is not directly
recoverable, and let $Y$ be the resulting multiset. The same lost $y$ may occur more than once. Thus $|Y|=q-\ell_x$, where cardinality is counted with multiplicity. Since a legal parity contains exactly one encodee from each source node and it already contains \(x\), every selected object \(y\in Y\) has a primary node different from node \(0\).

We use the following two facts implied by~\eqref{eq:Bij}. First, if an occurrence $y\in Y$ has primary node~$i$, then for each other failed node $j\ne i$, at most $\lambda$ parities of~$y$ also encode an object from~$j$. Second, for each
failed node $i\ne0$, at most $\lambda$ occurrences in~$Y$ have primary
node~$i$, because the parities of~$x$ from which these occurrences were selected are hosted at distinct nodes in~$B_{0i}$.

Consider the multiset of all pairs consisting of an occurrence in~$Y$ and a
parity that encodes its object:
\[
  \mathcal{Q} = \{(y,P) \mid y \in Y,\,
    P \text{ is a parity encoding } y\},
\]
where the multiplicity of each $y$ is inherited from~$Y$.
Since the object in every occurrence is encoded by exactly $q$ parities,
each occurrence contributes exactly $q$ pairs, and therefore
\[|\mathcal{Q}|=q(q-\ell_x).\]
For $s\in\{1,2,3\}$, let $\mathcal{Q}_s$ denote the pairs covered by
Case~$s$ below. Because every occurrence in~$Y$ is not directly recoverable,
each pair $(y,P)$ lies in at least one of these three submultisets.

\medskip\noindent\textbf{Case 1: $P$ is on failed node~$0$.}
Each occurrence contributes at most one such pair, so
$|\mathcal{Q}_1|\le q-\ell_x$.

\medskip\noindent\textbf{Case 2: $P$ encodes more than one lost encodee.}
For an occurrence on node~$i$, each of the other $p-1$ failed nodes can occur
with it in at most $\lambda$ parities. Hence
\[|\mathcal{Q}_2| \le \lambda(p-1)(q-\ell_x).\]

\medskip\noindent\textbf{Case 3: $P$ is lost and encodes exactly one lost object
(namely $y$).}
$P$ is hosted by a failed node $j\ne0$. Because every parity hosted at~$j$
contains one object from each source in $\mathcal S_{\mathrm{src}}(j)$,
the premise of this case implies that exactly one failed source, say~$i$,
sends to~$j$. Moreover, $j\notin L_x$, since every node in~$L_x$ receives
from failed node~$0$. There are $p-\ell_x-1$ possible storage nodes~$j$.
For each one, the second consequence above bounds the number of relevant
occurrences from its unique failed source by~$\lambda$. Therefore
\[|\mathcal{Q}_3| \le \lambda(p-\ell_x-1).\]

The number of pairs outside the union of the three cases is at least
\begin{align*}
N_{\mathrm{out}}
&:=q(q-\ell_x)-[1+\lambda(p-1)](q-\ell_x)
-\lambda(p-\ell_x-1)\notag\\
&=(q-\ell_x)[q-1-\lambda(p-1)]
-\lambda(p-\ell_x-1).
\end{align*}
Substitution of $q=\lambda(p-1)+2$ gives
\begin{equation*}
N_{\mathrm{out}}
=(q-\ell_x)-\lambda(p-\ell_x-1)
=(\lambda-1)\ell_x+2>0.
\end{equation*}
Thus some pair $(y^*,P^*)$ belongs to none of the three cases. Its parity is
alive and has $y^*$ as its only lost encodee, so $y^*$ is directly
recoverable, contrary to the construction of~$Y$. The desired parity~$P_0$
therefore exists. Recover its other lost encodees in the first round and~$x$
in the second. Since node~$0$ and object~$x$ were arbitrary, the conclusion
holds for every failed node and every lost object.
\end{proof}

\section{Refining the sufficient replica count}
\label{sec:lower-bounds}

Theorem~\ref{thm:main} gives a sufficient replica count for recovery,
uniformly over all legal placements and all failure sets of at most~$p$
nodes. However, it does not determine whether fewer replicas suffice for a
fixed zero-diagonal representative~$A$. We therefore introduce the universal
replica threshold~$q^*(A,p)$ and investigate how far the sufficient count can
be reduced. We first derive unconditional bounds and then identify incidence
structures that yield exact thresholds or sharper bounds.

\subsection{Threshold definition and unconditional bounds}

Assume $k\ge1$ for $p=1$ and
$k\ge\lambda(p-1)+2$ for $p\ge2$.

\begin{definition}[Universal replica threshold]\label{def:qstar}
For a fixed zero-diagonal representative~$A$ of an
$\SBIBD(v,k,\lambda)$ and an integer $p\ge1$,
let $q^*(A,p)$ be the smallest integer $q\in\{0,1,\cdots,k\}$ such that every
legal placement of any finite object collection using exactly $q$ target
nodes per object permits recovery from every set of at most $p$ failed nodes.
\end{definition}

By Theorem~\ref{thm:main}, the threshold $q^*(A,p)$ is well defined.

\begin{proposition}\label{prop:unconditional-lb}
For $p=1$, 
\begin{equation*}
q^*(A,1)=1.
\end{equation*}
For every $p\ge2$,
\begin{equation}\label{eq:qstar-general}
\max\{p,\lambda+1\}
\le q^*(A,p)
\le \lambda(p-1)+2.
\end{equation}
In particular,
\begin{equation}\label{eq:qstar-p2-interval}
\lambda+1\le q^*(A,2)\le\lambda+2.
\end{equation}
\end{proposition}

\begin{proof}
For $p=1$, failing the primary node of an object shows that $q=0$ does not
suffice, while Theorem~\ref{thm:main} shows that $q=1$ does. Hence
$q^*(A,1)=1$.

For $p\ge2$, first take $q=p-1$, select the $q$ target nodes of an object
$x$, and fail those nodes together with the primary node of~$x$. The primary
copy of~$x$ and every parity encoding~$x$ are lost. Hence $x$ cannot be
recovered, so $q^*(A,p)\ge p$.

Next choose two distinct primary nodes $i$ and $j$. Set $q=\lambda$ and
choose objects $x_i$ and $x_j$ with
$T_{x_i}=T_{x_j}=B_{ij}$, which is possible because $|B_{ij}|=\lambda$.
Fail only nodes $i$ and $j$, an admissible failure pattern for every
$p\ge2$. Since $A$ has zero diagonal,
$B_{ij}\cap\{i,j\}=\varnothing$, so all selected parities remain alive. At
every target in $B_{ij}$, place $x_i$ and $x_j$ in the same parity and
complete the remaining positions with objects whose primary nodes are alive.

These are the only parities containing $x_i$ or $x_j$: each object selects
exactly $q=\lambda$ targets, and each of its replicas is assigned to exactly
one parity. Every such parity contains both lost objects, so neither object is
uniquely recoverable. Thus
$q=\lambda$ does not guarantee recovery and
$q^*(A,p)\ge\lambda+1$. Combining the two constructions gives the lower
bound in~\eqref{eq:qstar-general}. Theorem~\ref{thm:main} gives the upper
bound. Finally, setting $p=2$
in~\eqref{eq:qstar-general} and using $\lambda\ge1$
gives~\eqref{eq:qstar-p2-interval}.
\end{proof}

\subsection{Exact classification for $p=2$}

The following condition distinguishes the two possible thresholds for a fixed
representative. Section~\ref{sec:reciprocal-realization} examines its behavior
under different zero-diagonal affinity labelings.

\begin{definition}[Reciprocal pair]\label{def:reciprocal-pair}
Two distinct nodes $i$ and $j$ form a \emph{reciprocal pair} if
\begin{equation*}
j\in\mathcal{S}_{\mathrm{tgt}}(i)
\quad\text{and}\quad
i\in\mathcal{S}_{\mathrm{tgt}}(j).
\end{equation*}
\end{definition}

\begin{proposition}[{Exact classification for $p=2$}]\label{prop:p2-exact}
Let $A$ be a fixed zero-diagonal representative of an
$\SBIBD(v,k,\lambda)$ and assume
$k\ge\lambda+2$. Then
\begin{equation*}
q^*(A,2)=
\begin{cases}
\lambda+2,& \text{if $A$ contains a reciprocal pair},\\
\lambda+1,& \text{otherwise}.
\end{cases}
\end{equation*}
\end{proposition}

\begin{proof}
First suppose that $i,j$ form a reciprocal pair, and let $B_{ij}$ be their
$\lambda$ common targets. Since the diagonal of~$A$ is zero, neither $i$
nor $j$ belongs to $B_{ij}$. Set $q=\lambda+1$. Choose an object $x_i$ on
node~$i$ with targets $B_{ij}\cup\{j\}$, and an object $x_j$ on node~$j$
with targets $B_{ij}\cup\{i\}$. Fail nodes $i$ and~$j$. The two parities
hosted on the failed nodes are lost. At every node in $B_{ij}$, arrange for
$x_i$ and $x_j$ to enter the same parity. After subtracting alive encodees,
all surviving equations have coefficient row $(1,1)$, so their rank is one.
Thus $q=\lambda+1$ does not guarantee recovery and
$q^*(A,2)=\lambda+2$ by~\eqref{eq:qstar-p2-interval}.

Now suppose that \(A\) has no reciprocal pair, and let \(i\) and \(j\) be the failed nodes. Without loss of generality, \(j\notin\mathcal{S}_{\mathrm{tgt}}(i)\). Hence all \(q=\lambda+1\) parities encoding an object \(x\) primary on \(i\) are alive. Only target nodes in \(B_{ij}\) can host such a parity containing another lost encodee primary on \(j\). Since \(q=\lambda+1>|B_{ij}|=\lambda\), at least one parity contains no other lost encodee, and \(x\) is directly recoverable. After recovering all objects primary on \(i\), each object primary on \(j\) has at least \(q-1=\lambda\) alive parities, and all other encodees in any such parity are known. Thus all objects primary on \(j\) are also recoverable. Therefore \(q=\lambda+1\) suffices, and combining with the lower bound in \eqref{eq:qstar-p2-interval} gives \(q^*(A,2)=\lambda+1\).
\end{proof}

Appendix~\ref{app:p2-example} gives an explicit zero-diagonal representative
of an $\SBIBD(7,4,2)$ exhibiting this obstruction from a reciprocal pair.

\subsection{ Classification for $p\ge 3$}

For $p\ge3$, the threshold depends on how the failed rows overlap at target
nodes. The following condition isolates the case with low overlap, and
Section~\ref{sec:separated-realization} studies when it exists and how it can
be constructed.

\begin{definition}[Pairwise separated set]
\label{def:pairwise-separated}
For $F\subseteq\{0,1,\cdots,v-1\}$ and
$t\in\{0,1,\cdots,v-1\}$, let
\begin{equation*}
d_F(t):=\bigl|\{i\in F:A_{it}=1\}\bigr|
=\sum_{i\in F}A_{it}.
\end{equation*}
The set $F$ is \emph{pairwise separated} if
\begin{equation*}
d_F(t)\le2 \text{ for every }t\in\{0,1,\cdots,v-1\}.
\end{equation*}
\end{definition}
The value $d_F(t)$ is the number of rows indexed by $F$ that have entry~$1$
in column~$t$. $F$ is pairwise separated if no column has
entry~$1$ in three distinct rows indexed by~$F$. Consequently, the sets
$B_{ij}$ corresponding to the unordered pairs
$\{i,j\}\subseteq F$ are mutually disjoint. For $i\in F$, define
\begin{equation*}
U_i(F):=
\bigcup_{j\in F\setminus\{i\}}B_{ij}.
\end{equation*}
If $F$ is pairwise separated and $|F|=p$, then
\[
|U_i(F)|=\lambda(p-1).
\]
When $\lambda=1$ and the design is a projective plane, this is the classical
arc condition that no three points are collinear \cite{kaplan2017}.

\begin{proposition}[Pairwise separated lower bound]
\label{prop:separated-lb}
Let $p\ge3$. If $A$ contains a pairwise separated set $F$ of size $p$ and
$k\ge\lambda(p-1)+2$, then $q=\lambda(p-1)$ does not guarantee recovery.
Consequently,
\begin{equation}\label{eq:separated-lb}
\lambda(p-1)+1
\le q^*(A,p)
\le\lambda(p-1)+2.
\end{equation}
\end{proposition}

\begin{proof}
Set $q=\lambda(p-1)$ and take the $p$ nodes in~$F$ as the failure set.
We construct a legal placement in which the selected objects are not uniquely
recoverable, showing that this value of~$q$ does not guarantee recovery.

For each $i\in F$, select an object $x_i$ primary on node~$i$ with
$T_{x_i}=U_i(F)$. Then 
\[|T_{x_i}| = \lambda(p-1) = q.\]
For every unordered pair $\{i,j\}\subseteq F$ and every
$t\in B_{ij}$, arrange the replicas of $x_i$ and $x_j$ to enter the same
parity hosted by node~$t$.

Consider any surviving parity $P$ encoding~$x_i$, and let $h$ denote its host.
Since $T_{x_i}=U_i(F)$, we have $h\in U_i(F)\setminus F$, so $h$ belongs to a
unique $B_{ij}$ with $j\in F\setminus\{i\}$. Membership in $B_{ij}$ gives
$d_F(h)\ge2$, while pairwise separation gives $d_F(h)\le2$, hence
$d_F(h)=2$. By construction, $x_i$ and $x_j$ are therefore the only lost
encodees in this parity. All other encodees are alive and can be cancelled,
leaving
\[
x_i\XOR x_j=\gamma_h,
\]
where $\gamma_h$ is known.

Let $w$ be any nonzero binary object of the same length as the selected
objects, and define
\[
x_i'=x_i\XOR w \text{ for every }i\in F.
\]
Every surviving parity contains either no selected object, or exactly two
selected objects due to $d_F(h) = 2$. For every surviving parity containing two selected objects,
\[
x_i'\XOR x_j'
=
(x_i\XOR w)\XOR(x_j\XOR w)
=
x_i\XOR x_j.
\]
Thus replacing $(x_i)_{i\in F}$ by $(x_i')_{i\in F}$ leaves every surviving
parity unchanged. Since $w\ne0$, these are two distinct object
collections that produce exactly the same surviving data. Hence the selected
objects are not uniquely recoverable when $q=\lambda(p-1)$. Therefore,
\[
q^*(A,p)\ge\lambda(p-1)+1.
\]
The upper bound follows from
Proposition~\ref{prop:unconditional-lb}, completing the proof of
\eqref{eq:separated-lb}.
\end{proof}

The preceding proposition identifies an obstruction supported by a
pairwise separated set of size~$p$. We now consider the complementary case in which
$A$ contains no such set.  Then every set of $p$ failed nodes has a
target shared by at least three of its nodes. Such higher overlap reduces the number of distinct targets blocked by other lost encodees and leads to a stronger recovery bound. Next we show that the resulting higher intersections refine the sufficient replica count by one.

\begin{proposition}[Recovery without pairwise separation]
\label{prop:no-separated-ub}
Let $p\ge3$ and $k\ge\lambda(p-1)+1$. If the fixed zero-diagonal
representative~$A$ contains no pairwise separated set of size~$p$, then
\begin{equation}\label{eq:no-separated-interval}
\max\{p,\lambda+1\}
\le q^*(A,p)
\le\lambda(p-1)+1.
\end{equation}
In particular, if $\lambda=1$, then
\begin{equation*}
q^*(A,p)=p.
\end{equation*}
\end{proposition}

\begin{proof}
Set
\[
q=\lambda(p-1)+1.
\]
We prove that every legal placement using $q$ target nodes per object is
recoverable from any set of at most $p$ failed nodes. It suffices to consider
exactly $p$ failed nodes. Fix such a failure set~$F$.

The recovery proceeds iteratively. Initially, let $H=F$. At each stage,
$H$ contains the failed nodes not yet processed by the recovery procedure.
We seek a node $i\in H$ such that every object primary on~$i$ has an alive
parity whose other encodees are already known. We then recover all objects
primary on~$i$, remove $i$ from~$H$, and repeat.

Fix a stage with $H\ne\varnothing$ and put $r=|H|$. For $i\in H$, a target
$t\in\mathcal{S}_{\mathrm{tgt}}(i)$ cannot provide the next recovery step
for objects primary on~$i$ for one of two reasons. First, every parity at~$t$
encoding an object primary on~$i$ also contains an encodee primary on another
node in~$H$. This occurs when $d_H(t)\ge2$, so let
\begin{equation*}
a_i(H):=
\left|
\left\{
t\in\mathcal{S}_{\mathrm{tgt}}(i):d_H(t)\ge2
\right\}
\right|.
\end{equation*}
Second, if $d_H(t)=1$, such a parity contains no encodee primary on another
node in~$H$, but it is still unavailable when its host~$t$ has failed. These
targets satisfy $t\in F$ and $d_H(t)=1$. Let
\begin{equation*}
b_i(H):=
\left|
\left\{
t\in F\cap\mathcal{S}_{\mathrm{tgt}}(i):d_H(t)=1
\right\}
\right|.
\end{equation*}
The two classes are disjoint. Hence the total number of unusable targets
in $\mathcal{S}_{\mathrm{tgt}}(i)$ at the current stage is
\begin{equation}\label{eq:ciH}
c_i(H):=a_i(H)+b_i(H).
\end{equation}
For simplicity, the dependence of $b_i(H)$ and $c_i(H)$ on the fixed
failure set~$F$ is omitted from the notation.

If $c_i(H)<q$, then the $q$-element target set~$ T_x\subseteq\mathcal{S}_{\mathrm{tgt}}(i)$ of any object~$x$
primary on node~$i$ cannot be contained in the set of unusable targets.
Hence some $t\in T_x$ is not counted by $c_i(H)$. By the definitions of
$a_i(H)$ and $b_i(H)$, we have $d_H(t)=1$ and $t\notin F$. Therefore, the
parity encoding~$x$ at node~$t$ is alive and contains no encodee primary
on another node in~$H$. All its other encodees are available from alive
primary nodes or have already been recovered. Thus every object primary
on~$i$ can be recovered at the current stage. It therefore suffices to show
that some $i\in H$ satisfies
\begin{equation}\label{eq:usable-node}
c_i(H)<q.
\end{equation}

We next estimate these two counts. Since every pair of distinct rows has
$\lambda$ common targets, $|B_{ij}|=\lambda$ for every
$j\in H\setminus\{i\}$, we have
\begin{align*}
\sum_{j\in H\setminus\{i\}}|B_{ij}| = \lambda(r-1).
\end{align*}
The sets $B_{ij}$ may not be disjoint. A target~$t$ shared by row~$i$
and $d_H(t)-1$ other rows in~$H$ is counted $d_H(t)-1$ times on the left hand side, but only once in $a_i(H)$. Define this excess multiplicity by
\begin{equation*}
\delta_i(H):=
\sum_{\substack{
t\in\mathcal{S}_{\mathrm{tgt}}(i), d_H(t)\ge2
}}
\bigl(d_H(t)-2\bigr),
\end{equation*}
which measures the excess multiplicity caused by targets shared by row~$i$
and at least two other rows in~$H$.
It follows that
\begin{align*}
a_i(H)+\delta_i(H)
&=
\sum_{\substack{
t\in\mathcal{S}_{\mathrm{tgt}}(i), d_H(t)\ge2
}}
\left[1+\bigl(d_H(t)-2\bigr)\right]
\notag\\
&=
\sum_{\substack{
t\in\mathcal{S}_{\mathrm{tgt}}(i), d_H(t)\ge2
}}
\bigl(d_H(t)-1\bigr)
\end{align*}

We count the pairs $(j,t)$
satisfying $j\in H\setminus\{i\}$ and $t\in B_{ij}$. For a fixed~$j$, there are $|B_{ij}|$ choices of~$t$. For a fixed~$t$, there are $d_H(t)-1$ choices of~$j$. It follows that 
\begin{equation*}
\begin{aligned}
\sum_{j\in H\setminus\{i\}} |B_{ij}|
&=
\sum_{t\in\mathcal{S}_{\mathrm{tgt}}(i)}
\bigl(d_H(t)-1\bigr)\\
&=
\sum_{\substack{
t\in\mathcal{S}_{\mathrm{tgt}}(i), d_H(t)\ge2
}}
\bigl(d_H(t)-1\bigr).
\end{aligned}
\end{equation*}
The final equality follows because a target with $d_H(t)=1$ contributes zero to the sum. Combining the preceding identities gives
\begin{align}
a_i(H)+\delta_i(H) = \lambda(r-1).
\label{eq:block_identity_final}
\end{align}

Moreover, every failed target with $d_H(t)=1$ is counted by $b_i(H)$ for
exactly one node $i\in H$. Therefore
\begin{equation}\label{eq:biH-sum}
\sum_{i\in H}b_i(H)
=
\left|
\left\{
t\in F:d_H(t)=1
\right\}
\right|
\le p.
\end{equation}

Suppose otherwise that $c_i(H)\ge q$ for every $i\in H$. Equations
\eqref{eq:ciH} and~\eqref{eq:block_identity_final} give
\begin{align}
b_i(H)
&\ge q-a_i(H)
\nonumber\\
&=
\lambda(p-1)+1
-\bigl(\lambda(r-1)-\delta_i(H)\bigr)
\nonumber\\
&=
\lambda(p-r)+1+\delta_i(H).
\label{eq:biH-lower}
\end{align}

First suppose that $r=p$. Then $H=F$. Since $A$ contains no
pairwise separated set of size~$p$, some target $t$ satisfies
$d_H(t)\ge3$. Consequently, $\delta_i(H)>0$ for at least one node $i\in H$.
Summing \eqref{eq:biH-lower} over $H$ gives
\[
\sum_{i\in H}b_i(H)
\ge
p+\sum_{i\in H}\delta_i(H)
>p,
\]
contradicting \eqref{eq:biH-sum}.

Next suppose that $2\le r<p$. Since $\lambda\ge1$, summing
\eqref{eq:biH-lower} gives
\begin{align*}
\sum_{i\in H}b_i(H)
&\ge
r\bigl[\lambda(p-r)+1\bigr]
\\
&\ge
r(p-r+1)
\\
&>p,
\end{align*}
where the last inequality follows from
\[
r(p-r+1)-p=(r-1)(p-r)>0.
\]
This again contradicts \eqref{eq:biH-sum}.

Finally, if $r=1$, then $a_i(H)=0$. Since $A_{ii}=0$, at most the other
$p-1$ failed nodes can be targets of node $i$, and hence
\[
c_i(H)=b_i(H)\le p-1<q.
\]
Thus \eqref{eq:usable-node} holds in every case.

Choose such an $i$ satisfying \eqref{eq:usable-node}, recover all objects primary on~$i$, remove $i$ from~$H$,
and repeat. At least one node is removed at each stage, so after at most $p$
stages, $H=\varnothing$ and all lost objects have been recovered.
This proves
\[
q^*(A,p)\le\lambda(p-1)+1.
\]

The constructions for the lower bound used in the proof of
Proposition~\ref{prop:unconditional-lb} require only the replica counts
$p-1$ and $\lambda$, both of which are feasible under
$k\ge\lambda(p-1)+1$. They therefore give
\[
q^*(A,p)\ge\max\{p,\lambda+1\}.
\]
This proves \eqref{eq:no-separated-interval}. When $\lambda=1$, the two
bounds coincide and give
\[
q^*(A,p)=p.
\]
\end{proof}

For $\lambda>1$, the exact value of $q^*(A,p)$ within this interval depends
on finer local incidence information that is not captured by the SBIBD
parameters or by the existence of a pairwise separated set alone. Sharper
bounds for a specific representative might be obtained by introducing additional conditions
that record intersections of multiplicity at least three and their positions relative to
failed parity hosts. However, such refinements would produce a conclusion that increasingly depends on the individual case. We therefore do not refine the case further. 

We now return
to representatives that contain pairwise separated sets.
Pairwise separation alone leaves a gap of one between the bounds in
\eqref{eq:separated-lb}. The following additional condition closes that gap.
Section~\ref{sec:cycle-realization} shows how it can be realized by
choosing an affinity labeling.

\begin{definition}[Private target permutation]
\label{def:private-target-permutation}
Let $F$ be a pairwise separated set with $|F|=p$. We say that $F$ admits a \emph{private target permutation} if there exists a permutation
$\sigma:F\to F$ such that
\begin{equation*}
\sigma(i)\in\mathcal{S}_{\mathrm{tgt}}(i)
\text{ and }
d_F\bigl(\sigma(i)\bigr)=1
\text{ for every }i\in F.
\end{equation*}
\end{definition}

The definition assigns each node~$i\in F$ a private target
$\sigma(i)\in F$. These targets are all distinct because $\sigma$ is a
permutation. Since $A$ has zero diagonal, $\sigma(i)\ne i$. 

For $p=2$, pairwise separation is automatic and  a private target
permutation exists exactly when the two nodes form a reciprocal pair. For $p\ge4$, a private target permutation need not consist of a
single cycle. For example, when $p=4$, it may consist of one cycle of length
four, $(i_1\,i_2\,i_3\,i_4)$, or two disjoint cycles of length two,
$(i_1\,i_2)(i_3\,i_4)$.

\begin{proposition}[Exact threshold under pairwise separation]
\label{prop:private-target-exact}
Let $p\ge3$ and $k\ge\lambda(p-1)+2$. Suppose that $A$ contains at least
one pairwise separated set of size~$p$. Then $q^*(A,p)$ depends on whether any such set admits a private target
permutation:
\begin{equation}
\label{eq:private-target-exact}
q^*(A,p)=
\begin{cases}
\lambda(p-1)+2, & \text{if at least one does},\\
\lambda(p-1)+1, & \text{otherwise}.
\end{cases}
\end{equation}
\end{proposition}

\begin{proof}
If a pairwise separated
set admits a private target permutation, we extend the nonrecoverable
placement from Proposition~\ref{prop:separated-lb} by assigning one additional
failed target to each selected object. The additional parities are lost, so
the original ambiguity remains. Conversely, if no such permutation exists, we reuse the iterative recovery
framework and blocked target counts from Proposition~\ref{prop:no-separated-ub}.
At every later stage, these counts identify a node whose primary objects can
all be recovered. The only structure that could prevent the initial recovery step would itself define a private target permutation.

First, let $\sigma$ be a private target permutation of a
pairwise separated set~$F$ with $|F|=p$. For each $i\in F$,  following the similar construction used in the proof of Proposition \ref{prop:separated-lb}, choose an object $x_i$ primary on node~$i$ and set
\[
T_{x_i}
=U_i(F)\cup\{\sigma(i)\}.
\]
Every $t\in U_i(F)$ satisfies $d_F(t)=2$, whereas
$d_F\bigl(\sigma(i)\bigr)=1$. Hence
$\sigma(i)\notin U_i(F)$ and
\[
|T_{x_i}|=|U_i(F)|+1=\lambda(p-1)+1.
\]

For every unordered pair $\{i,j\}\subseteq F$ and every $t\in B_{ij}$,
place $x_i$ and $x_j$ in the same parity at node~$t$. Complete the remaining
parity positions with objects whose primary nodes lie outside~$F$ to obtain
a legal placement. 

Fail all nodes in~$F$. Since $\sigma(i)\in F$, the parity hosted at the
added target $\sigma(i)$ is lost.
Every surviving parity encoding some $x_i\ (i\in F)$ is hosted at a
target $t\in B_{ij}$ for a unique $j\in F\setminus\{i\}$ and contains
exactly the two lost objects $x_i$ and~$x_j$. Consequently, replacing every $x_i$
by $x_i\XOR w$, for any fixed nonzero binary word~$w$, leaves all surviving
parity equations unchanged. Thus $q=\lambda(p-1)+1$ does not guarantee unique
recovery, and hence
\[
q^*(A,p)\ge\lambda(p-1)+2.
\]
Theorem~\ref{thm:main} gives the reverse inequality, proving the first case
of~\eqref{eq:private-target-exact}.

Conversely, suppose that no pairwise separated set of size~$p$ admits a
private target permutation, and set
\[
q=\lambda(p-1)+1.
\]
Enlarge any failure set to a set~$F$ of $p$ nodes. Use the notation for each recovery stage and counting identities from the proof of
Proposition~\ref{prop:no-separated-ub}, which apply to every failure
set~$F$.

At a recovery stage $H\subseteq F$, first suppose that
$2\le r=|H|<p$. If $c_i(H)\ge q$ for every $i\in H$, then
\eqref{eq:block_identity_final} gives
\[
p
\ge
\sum_{i\in H}b_i(H)
\ge
r\bigl[\lambda(p-r)+1\bigr]
\ge
r(p-r+1)
>
p,
\]
a contradiction. If $r=1$, the zero diagonal gives
\[
c_i(H)=b_i(H)\le p-1<q.
\]

It remains to consider the initial stage $H=F$. Suppose that
$c_i(F)\ge q$ for every $i\in F$. Equation~\eqref{eq:block_identity_final} 
then gives
\[
b_i(F)\ge1+\delta_i(F) \text{ for every }i\in F.
\]
Together with~\eqref{eq:biH-sum}, this forces
\[
\delta_i(F)=0,\ b_i(F)=1 \text{ for every }i\in F,
\]
and
\[
\sum_{i\in F}b_i(F)=p.
\]
The equalities $\delta_i(F)=0$ show that $F$ is pairwise separated.

Moreover, for every $i\in F$, there is exactly one $ \sigma(i)\in F\cap\mathcal{S}_{\mathrm{tgt}}(i)$
such that $d_F\bigl(\sigma(i)\bigr)=1$.
These values are distinct because a column with incidence degree one
cannot belong to $\mathcal{S}_{\mathrm{tgt}}(i)$ for two different
$i\in F$. Hence $\sigma:F\to F$ is a private target permutation. This contradicts the hypothesis.

Therefore, at every recovery stage, some $i\in H$ satisfies $c_i(H)<q$.
All objects primary on~$i$ can be recovered, after which $i$ is removed
from~$H$. Repeating this step recovers every lost object. Hence
\[
q^*(A,p)\le\lambda(p-1)+1.
\]
Proposition~\ref{prop:separated-lb} gives the reverse inequality, proving
the second case of~\eqref{eq:private-target-exact}.

\end{proof}

Appendix~\ref{app:p3-example} gives an explicit $\SBIBD(16,6,2)$ example
with $p=3$ in which a pairwise separated set admits a private target
permutation.

\subsection{Summary of thresholds for fixed representatives}

%

%
Table~\ref{tab:threshold-summary} summarizes the exact values and structural bounds established above. Its two rows for $p=2$ are distinguished by the
existence of a reciprocal pair. For $p\ge3$, it gives a
mutually exclusive and exhaustive classification according to the
existence of pairwise separated sets of size~$p$, parameter~$\lambda$, and whether any such set
admits a private target permutation.

\begin{table}[htbp]
\caption{Exact values and bounds for the universal replica threshold of a
fixed zero-diagonal SBIBD representative~$A$. The standing feasibility assumptions
are $k\ge1$ for $p=1$ and $k\ge\lambda(p-1)+2$ for $p\ge2$.}
\label{tab:threshold-summary}
\begin{tabularx}{\linewidth}{@{}
>{\raggedright\arraybackslash}X
>{\raggedright\arraybackslash}X
@{}}
\toprule
Failure count $p$ and structural conditions on $A$ & Threshold or current bounds \\
\midrule
$p=1$
  & $q^*(A,1)=1$ \\
\addlinespace[0.8ex]
$p=2$, no reciprocal pair
  & $q^*(A,2)=\lambda+1$ \\

$p=2$, a reciprocal pair exists
  & $q^*(A,2)=\lambda+2$ \\
\addlinespace[0.8ex]
\multicolumn{2}{@{}l}{%
{$p\ge3$, no pairwise separated set of size~$p$}}\\

\hspace{4em}$\lambda=1$
  & $q^*(A,p)=p$ \\

\hspace{4em}{$\lambda>1$}
  & {$\max\{p,\lambda+1\}\le q^*(A,p)
    \le\lambda(p-1)+1$} \\
    
\addlinespace[0.8ex]
\multicolumn{2}{@{}l}{%
{$p\ge3$, a pairwise separated set of size~$p$ exists}}\\[-0.3ex]

\hspace*{4em}No set admits a private target permutation
  & $q^*(A,p)=\lambda(p-1)+1$ \\

\hspace*{4em}Some set admits a private target permutation
  & $q^*(A,p)=\lambda(p-1)+2$ \\
\bottomrule
\end{tabularx}
\end{table}

The exact threshold for a fixed representative remains open when $p\ge3$,
$\lambda>1$, and $A$ contains no pairwise separated set of size~$p$. Appendix~\ref{app:p2-example} gives an explicit zero-diagonal
$\SBIBD(7,4,2)$ representative that realizes the case with a reciprocal pair for
$p=2$. 
Appendices~\ref{app:p3-example} and \ref{app:p3-no-private-example} give two zero-diagonal representatives of the
same $\SBIBD(16,6,2)$ for $p=3$. In one representative, some
pairwise separated set of size~$3$ admits a private target permutation. In the
other, pairwise separated sets of size~$3$ exist, but none admits a private
target permutation.

\section{Existence and relabeling of threshold structures}
\label{sec:realization}
Section~\ref{sec:lower-bounds} identified reciprocal pairs, pairwise separated sets and private target permutations as the structures controlling the replica threshold. Pairwise separation depends only on the underlying SBIBD, whereas a private target permutation may also depend on the affinity labeling.  For $p=2$, pairwise separation is automatic, and a private target
permutation is precisely a reciprocal pair. 

This section first studies the existence of pairwise separated sets and then shows how suitable
zero-diagonal labelings realize private target permutations.

\subsection{Existence and construction of pairwise separated sets}
\label{sec:separated-realization}
The following condition depending only on the parameters gives a constructive guarantee for
the existence of a pairwise separated set of size~$p$.

\begin{proposition}[Sufficient condition for pairwise separation]
\label{prop:separated-existence}
Let $A$ be any representative of an $\SBIBD(v,k,\lambda)$ and let $p\ge3$.
If
\begin{equation}\label{eq:separated-sufficient}
v>(p-1)+\lambda\cdot \binom{p-1}{2}\cdot (k-2),
\end{equation}
then $A$ contains a pairwise separated set of size~$p$.
For $p=3$, the condition becomes
\[
\lambda(k-2)<v-2.
\]
\end{proposition}

\begin{proof}
Start with any two row indices, which automatically form a pairwise separated
set. Suppose that a pairwise separated set~$F$ with
$|F|=m<p$ has been selected. The sets $B_{ij}$ over all unordered pairs
$\{i,j\}\subseteq F$ are mutually disjoint. Since $|B_{ij}|=\lambda$, there are exactly
\[
\lambda\cdot \binom{m}{2}
\]
column indices~$t$ satisfying $d_F(t)=2$.

For each such~$t$, exactly $k-2$ row indices $u\notin F$ satisfy
$A_{ut}=1$. None of these indices can be added to $F$, since doing so would give
$d_{F\cup\{u\}}(t)=3$. Thus at most
\[
\lambda\binom{m}{2}(k-2)
\]
indices outside $F$ are excluded. A new index can therefore be selected whenever
\[
v>m+\lambda\binom{m}{2}(k-2).
\]
The right hand side increases with $m$. Condition
\eqref{eq:separated-sufficient}, which is the case $m=p-1$, guarantees that
the construction continues until $|F|=p$.

The proof also gives a greedy construction. At each step, choose any
$u\notin F$ satisfying
\[
A_{ut}=0, \text{ for every $t$ such that $d_F(t)=2$}.
\]
\end{proof}

Pairwise separated sets are not automatic, even under the feasibility
condition $k\ge\lambda(p-1)+2$. For any prime power $\rho\ge3$, incidence between points and planes in $\mathrm{PG}(3,\rho)$ gives a symmetric design with parameters
\[
(v,k,\lambda)=(\rho^3+\rho^2+\rho+1,\rho^2+\rho+1,\rho+1).
\]
Any three points lie in a plane, so no three rows are pairwise separated.
Moreover, $k\ge2\lambda+2$ for $\rho\ge3$, making this a feasible family for
$p=3$. The smallest member has parameters $(40,13,4)$~\cite{colbourn2010}.
Since the absence of a pairwise separated set of size~$3$ excludes every
larger such set, the same family is feasible for every $3\le p\le\rho$, because
$k\ge\lambda(p-1)+2$ in precisely that range.
Equivalently, whenever $\lambda-1$ is a prime power and
$3\le p\le\lambda-1$, this construction gives an explicit design with
$k=\lambda^2-\lambda+1$ that contains no pairwise separated set.

\begin{corollary}[Pairwise separated sets for large $k$]
\label{cor:separated-asymptotic}
Fix $p\ge3$ and $\lambda\ge1$. There exists $K=K(p,\lambda)$ such that,
whenever $k\ge K$ and an $\SBIBD(v,k,\lambda)$ exists, every such design
contains a pairwise separated set of size~$p$.
\end{corollary}

\begin{proof}
Substituting $v=k(k-1)/\lambda+1$ into
\eqref{eq:separated-sufficient} gives
\[
\frac{k(k-1)}{\lambda}+1
>
(p-1)+\lambda\binom{p-1}{2}(k-2).
\]
For fixed $p$ and $\lambda$, the left hand side is quadratic in $k$ and the
right hand side is linear in $k$. Hence the inequality holds for all
sufficiently large feasible values of $k$.
\end{proof}

Consequently, for fixed $p$ and $\lambda$, a design without a
pairwise separated set of size~$p$ can occur for finitely many possible values of $k$.

\subsection{Private target permutations under zero-diagonal relabeling}
\label{sec:reciprocal-realization}
\label{sec:cycle-realization}

\begin{proposition}[Realizing private target permutations]
\label{prop:private-target-relabel}
Suppose $p\ge2$ and $k\ge\lambda(p-1)+2$. Let $A$ be a representative of an $\SBIBD(v,k,\lambda)$. If $A$ contains a pairwise separated set~$F$ of size~$p$, then the columns of $A$ can be relabeled to obtain a zero-diagonal
representative~$A'$ in which $F$ admits a private target permutation.
Consequently,
\begin{equation}\label{eq:private-target-relabel}
q^*(A',p)=\lambda(p-1)+2.
\end{equation}
\end{proposition}

\begin{proof}
We treat $p=2$ and $p\ge3$ separately. 

First let $p=2$.  Let $F=\{i,j\} \subseteq \{0,1,\cdots,v-1\}$ be an arbitrary set of size two. Since $d_F(t)\le2$ for every column~$t$, the set~$F$ is automatically pairwise separated. 

Construct a bipartite graph whose two vertex parts are the row indices $\{0,1,\cdots,v-1\}$ and column indices $\{0,1,\cdots,v-1\}$. Connect a row vertex $h$ and a column vertex $t$ if and only if $A_{ht}=0$. Since every row and column of~$A$ contains exactly $v-k$ zero entries, every vertex has degree
\[
d:= v-k=\frac{(k-1)(k-\lambda)}{\lambda}.
\] 
The assumption $k\ge\lambda+2$ implies that $d>2$. Since $d$ is an
integer, we have $d\ge3$. The graph is the incidence graph of the complementary
symmetric design, so it is bipartite and distance regular.

Since
\[
\bigl|
\mathcal{S}_{\mathrm{tgt}}(i)
\setminus
\mathcal{S}_{\mathrm{tgt}}(j)
\bigr|
=
\bigl|
\mathcal{S}_{\mathrm{tgt}}(j)
\setminus
\mathcal{S}_{\mathrm{tgt}}(i)
\bigr|
=
k-\lambda\ge2,
\]
there exist two distinct 
columns $\tau_i$ and $\tau_j$ such that
\[
\tau_i\in
\mathcal{S}_{\mathrm{tgt}}(j)
\setminus
\mathcal{S}_{\mathrm{tgt}}(i)
\]
and
\[
\tau_j\in
\mathcal{S}_{\mathrm{tgt}}(i)
\setminus
\mathcal{S}_{\mathrm{tgt}}(j).
\]
By their definitions,
\[
A_{i\tau_i}=A_{j\tau_j}=0,\ A_{j\tau_i}=A_{i\tau_j}=1.
\]

The two zero entries at $(i,\tau_i)$ and $(j,\tau_j)$ lie in different
rows and different columns in $A$. Hence the corresponding edges $(i,\tau_i)$ and $(j,\tau_j)$ in the graph have no common endpoint. By
\cite[Theorem~3.17]{cioaba2017}, in a bipartite graph that is distance regular
of degree~$d$, every set of at most
$\lfloor(d+1)/2\rfloor$ edges with no common endpoint is contained in
a perfect matching. Since $d\ge3$, we have
$\lfloor(d+1)/2\rfloor\ge2$. Hence some perfect matching contains both $(i,\tau_i)$ and $(j,\tau_j)$.

A perfect matching pairs every row vertex with a distinct column vertex. Consequently, it defines a bijection
\[
\pi: \{0,1,\cdots,v-1\}\longrightarrow \{0,1,\cdots,v-1\}
\]
such that
\[
A_{h,\pi(h)}=0, \text{ for every }h\in\{0,1,\cdots,v-1\},
\]
with
\[
\pi(i)=\tau_i,\ \pi(j)=\tau_j.
\]

Define the relabeled representative~$A'$
by
\[
A'_{rh}:=A_{r,\pi(h)},\ r,h\in\{0,1,\cdots,v-1\},
\]
then
\[
A'_{hh}=A_{h,\pi(h)}=0,\ h\in\{0,1,\cdots,v-1\}.
\]
$A'$ is zero-diagonal. Since $\pi(i)=\tau_i$ and
$\pi(j)=\tau_j$,
\[
A'_{ji}=A_{j,\pi(i)}=A_{j\tau_i}=1,\ A'_{ij}=A_{i,\pi(j)}=A_{i\tau_j}=1,
\]
we have that $F=\{i,j\}$  is a reciprocal pair in~$A'$. Define
\[
\sigma(i)=j,\ \sigma(j)=i,
\]
then $\sigma$ is a private target permutation of~$F$ in~$A'$.

Now let $p\ge3$, and write the 
pairwise separated set as
\[
F=\{i_1,i_2,\cdots,i_p\}.
\]
For each fixed $r\in\{1,2,\cdots,p\}$, 
\[
\mathcal{S}_{\mathrm{tgt}}(i_r)\setminus U_{i_r}(F)
=
\bigl\{
t\in\mathcal{S}_{\mathrm{tgt}}(i_r):d_F(t)=1
\bigr\}.
\]
Because
\[
\bigl|\mathcal{S}_{\mathrm{tgt}}(i_r)\bigr|=k,\ |U_{i_r}(F)| = \lambda(p-1),
\]
we have
\[
\bigl|\bigl\{
t\in\mathcal{S}_{\mathrm{tgt}}(i_r):d_F(t)=1
\bigr\}\bigr| =k-\lambda(p-1)\ge2.
\]
Choose one such column and denote it by~$t_r$, i.e.,
\[
t_r\in\mathcal{S}_{\mathrm{tgt}}(i_r),\ d_F(t_r)=1.
\]
Repeating this choice for
$r\in\{1,2,\cdots,p\}$ gives a family of columns
$t_1,t_2,\cdots,t_p$. These columns are distinct. Otherwise, if
$t_r=t_s$ for $r\ne s$, the common column would belong to both
$\mathcal{S}_{\mathrm{tgt}}(i_r)$ and
$\mathcal{S}_{\mathrm{tgt}}(i_s)$, giving $d_F(t_r)\ge2$.

Define $i_{p+1}=i_1$. Since $i_{r+1}\ne i_r$ and $d_F(t_r)=1$, we have
\[
t_r\notin\mathcal{S}_{\mathrm{tgt}}(i_{r+1}),
\]
and hence
\[
A_{i_{r+1},t_r}=0.
\]
The row indices $i_{r+1}$ are distinct, and the column indices $t_r$
are distinct. Hence the $p$ zero entries $A_{i_{r+1},t_r},\ r\in\{1,2,\cdots,p\}$ lie in different rows and different columns.

Delete the rows in~$F$ and the columns
$t_1,t_2,\cdots,t_p$ from~$A$. The resulting submatrix has
$v-p$ rows and $v-p$ columns. Since every row and column of~$A$
contains $v-k$ zero entries, every row and column of the submatrix
contains at least $v-k-p$ zero entries. The assumption $k\ge\lambda(p-1)+2$ gives
\[
\lambda\le\frac{k-2}{p-1}.
\]
Using the SBIBD identity, we obtain
\[
v=\frac{k(k-1)}{\lambda}+1
\ge
\frac{(p-1)k(k-1)}{k-2}+1
\ge
2k+p.
\]
The last inequality follows from
\[
\frac{(p-1)k(k-1)}{k-2}+1-(2k+p)
=
(p-3)k+\frac{2(p-1)}{k-2}
\ge0.
\]
Therefore,
\[
v-k-p\ge\frac{v-p}{2}.
\]

Let $R_0:=\{0,1,\cdots,v-1\}\setminus F$ and $C_0:= \{0,1,\cdots,v-1\}\setminus \{t_1,t_2,\cdots,t_p\}.$
Construct a bipartite graph with vertex parts $R_0$ and~$C_0$, in which
row vertex $h$ and column vertex~$t$ are connected if and only if $A_{ht}=0$.
Both parts have size
\[
n:=v-p,
\]
and every vertex has degree at least
\[
v-k-p\ge\frac n2.
\]

We verify Hall's condition. Let $S\subseteq R_0$, let $N(S) \subseteq  C_0$ be the set of column indices
adjacent to at least one element of $S$. If
$1\le |S|\le n/2$, then the neighbors of any row vertex in~$S$
already give
\[
|S| \le \frac n2\le |N(S)|.
\]
If $|S|>n/2$, suppose that some column vertex lies outside~$N(S)$.
All its neighbors would then lie in $R_0\setminus S$, which contains
fewer than $n/2$ vertices. This contradicts its degree being at least
$n/2$. Hence $N(S)=C_0$, and 
\[
|S| \le n = |N(S)|.
\]
Hall's theorem therefore gives a perfect matching between $R_0$
and~$C_0$. Consequently, there is a bijection
\[
\pi_0:R_0\longrightarrow C_0
\]
such that
\[
A_{h,\pi_0(h)}=0 \text{ for every }h\in R_0.
\]

Extend $\pi_0$ to a bijection
\[
\pi:\{0,1,\cdots,v-1\}\longrightarrow
\{0,1,\cdots,v-1\}
\]
by setting
\[
\pi(i_{r+1})=t_r\ (r\in\{1,2,\cdots,p\}).
\]
Then
\[
A_{h,\pi(h)}=0 \text{ for every }h\in\{0,1,\cdots,v-1\}.
\]
Keep the row labels fixed and define the relabeled representative~$A'$
by
\[
A'_{h\ell}:=A_{h,\pi(\ell)}\ (h,\ell\in\{0,1,\cdots,v-1\}).
\]
It follows that
\[
A'_{hh}=A_{h,\pi(h)}=0,
\]
so $A'$ is zero-diagonal. Moreover,  for $r\in\{1,2,\cdots,p\}$,
\[
A'_{i_r,i_{r+1}}
=
A_{i_r,\pi(i_{r+1})}
=
A_{i_r,t_r}
=
1
\]
and
\[
\sum_{s=1}^{p}A'_{i_s,i_{r+1}}
=
\sum_{s=1}^{p}A_{i_s,t_r}
=
1.
\]

The column relabeling also preserves pairwise separation because
\[
\sum_{i\in F}A'_{i\ell}
=
\sum_{i\in F}A_{i,\pi(\ell)}
\le2, \text{ for $\ell\in\{0,1,\cdots,v-1\}$}.
\]
 Therefore, the cyclic
permutation
\[
\sigma(i_r)=i_{r+1},\ r\in\{1,2,\cdots,p\}
\]
is a private target permutation of~$F$ in~$A'$.

For $p=2$, Proposition~\ref{prop:p2-exact} gives
\eqref{eq:private-target-relabel}. For $p\ge3$,
Proposition~\ref{prop:private-target-exact} gives the same conclusion.
\end{proof}

For $p=2$, Proposition~\ref{prop:private-target-relabel} shows that
every underlying SBIBD admits at least one zero-diagonal representative
containing a reciprocal pair. It does not assert that every zero-diagonal
representative contains such a pair. For example, the cyclic representative~$A$ of an
$\SBIBD(7,3,1)$ defined by
\[
A_{ij}=1 \Longleftrightarrow j-i\pmod 7\in\{1,2,4\}
\]
has zero diagonal but no reciprocal pair, because
$\{1,2,4\}\cap\{-1,-2,-4\}=\varnothing$ modulo~$7$. Hence this representative
has threshold $q^*(A,2)=2=\lambda+1$, although another zero-diagonal affinity
labeling of the same $\SBIBD$ yields a representative with threshold~$3$.

Proposition~\ref{prop:private-target-relabel} concerns the existence of a
zero-diagonal affinity labeling whose representative attains the upper
endpoint. Changing the point and block labels independently preserves the
SBIBD but need not preserve the zero diagonal. A simultaneous node relabeling
$A\mapsto \Pi A\Pi^{\mathsf T}$, where $\Pi$ is a permutation matrix,
preserves both the zero diagonal and the existence or nonexistence of a
private target permutation, so it cannot create one in a fixed representative.
The matching construction instead chooses a different affinity labeling whose
representative remains zero-diagonal.

Combining Corollary~\ref{cor:separated-asymptotic} and
Proposition~\ref{prop:private-target-relabel}, for fixed $p\ge3$ and~$\lambda$
and all sufficiently large feasible~$k$, every existing SBIBD admits a
zero-diagonal affinity labeling whose representative~$A'$ satisfies
$q^*(A',p)=\lambda(p-1)+2$. This is still not a statement that every fixed
zero-diagonal representative attains that endpoint. The threshold therefore depends on local incidence structures that are not
determined by the parameters $(v,k,\lambda)$ alone. 


\section{Discussion}\label{sec:disc}

This section discusses three aspects of the generalized stripeless erasure coding scheme studied in this paper, i.e., parameter flexibility relative to Nos, the selection
of SBIBDs and their affinity labelings, and the resulting failure tolerance
and redundancy. Compared to Nos, our construction provides more flexible
parameter choices. For a suitable SBIBD representative~$A$, the required replica count can be reduced by
one.

\subsection{Parameter flexibility beyond Nos}

Nos uses an SBIBD with $\lambda=1$~\cite{nos2025}. The SBIBD identity then
reduces to
\[
v=k(k-1)+1,
\]
which restricts the admissible combinations of the cluster size~$v$ and the
number~$k$ of admissible target nodes. Theorem~\ref{thm:main} confirms the
Nos sufficient count
\[
q=p+1
\]
for $p\ge2$ under $k\ge p+1$. For $p=1$, the exact count is $q=1$.

Allowing any admissible value of~$\lambda$ changes the SBIBD identity to
\[
v=\frac{k(k-1)}{\lambda}+1.
\]
This provides more flexibility in choosing the cluster size $v$ and the number of
admissible target nodes $k$. For example, an $\SBIBD(11,5,2)$ exists, whereas no
$\SBIBD(11,k,1)$ exists because $k^2-k+1=11$ has no integral solution
\cite{colbourn2010}. For $p=2$, Theorem~\ref{thm:main} applies to the
former design with $q=4$.

This additional parameter flexibility involves a tradeoff. For fixed~$k$, a
larger value of~$\lambda$ reduces the value of~$v$ permitted by the design
identity, but increases the sufficient replica count
\[
q=\lambda(p-1)+2.
\]
For prescribed $v$ and~$p$, the parameters must therefore satisfy
\begin{equation*}
v=\frac{k(k-1)}{\lambda}+1,\ k\ge q,
\end{equation*}
and an SBIBD with these parameters must exist. 

\subsection{Selecting a design and an affinity labeling}

The results in Sections~\ref{sec:lower-bounds} and
\ref{sec:realization} provide criteria for selecting an SBIBD and its affinity
labeling when the goal is to minimize the replica threshold.
The selection has two levels. The parameters and the underlying block design
determine the admissible incidence structure. The affinity labeling then
determines which points and blocks represent the same storage nodes.
Consequently, matrices obtained from different zero-diagonal affinity
labelings of the same SBIBD may have different replica thresholds.

For $p=2$, a zero-diagonal representative without a reciprocal pair is
preferable whenever one exists. Proposition~\ref{prop:p2-exact} gives
\[
q^*(A,2)=\lambda+1
\]
for such a representative, compared with
\[
q^*(A,2)=\lambda+2
\]
when a reciprocal pair exists.

For $p\ge3$, the structure that attains the sufficient endpoint is a
pairwise separated set of size~$p$ that admits a private target permutation.
If some such set exists, Proposition~\ref{prop:private-target-exact} gives
\[
q^*(A,p)=\lambda(p-1)+2.
\]
If pairwise separated sets of size~$p$ exist but none admits a private target
permutation, then
\[
q^*(A,p)=\lambda(p-1)+1.
\]
Choosing an underlying design with no pairwise separated set of size~$p$ is
another sufficient way to reduce the replica thresholds. In this case,
Proposition~\ref{prop:no-separated-ub} gives
\[
q^*(A,p)\le\lambda(p-1)+1.
\]
When $\lambda=1$, this bound is exact and gives $q^*(A,p)=p$.

In particular, Nos uses $\lambda=1$ and the general sufficient replica
count $q=p+1$ for $p\ge2$. Whenever a suitable zero-diagonal
representative is available, the structural refinements developed above
reduce this count by one, from $p+1$ to~$p$. For $p=2$, this reduction is
obtained by choosing a representative without a reciprocal pair. For
$p\ge3$, it is obtained by choosing a representative in which no
pairwise separated set of size~$p$ admits a private target permutation.
An underlying design with no pairwise separated set of size~$p$ is another
way to ensure the latter condition. In either case, the amortized parity
storage per primary object decreases from $(p+1)/k$ to $p/k$.

\subsection{Failure tolerance and redundancy}
For fixed SBIBD parameters $(v, k, \lambda)$, the largest number of node failures certified by Theorem~\ref{thm:main} is
\begin{equation*}
p_{\mathrm{universal}}
=
\max\left\{
1,
\left\lfloor\frac{k-2}{\lambda}\right\rfloor+1
\right\}.
\end{equation*}
Each data object contributes to~$q$ parities, and every parity combines~$k$
objects. Since the temporary replicas are deleted after parity formation, the
amortized persistent parity storage per primary object is
\[
\frac{q}{k}.
\]
Including the primary copy gives the total storage ratio
\[
1+\frac{q}{k}.
\]
The transmission cost of encoding is~$q$ replicas per object.

Theorem~\ref{thm:main} gives the universally sufficient ratio of parity storage
\begin{equation*}
\eta_{\mathrm{suff}}=
\begin{cases}
1/k, & p=1,\\
[\lambda(p-1)+2]/k, & p\ge2.
\end{cases}
\end{equation*}
For a fixed representative $A$ and a failure tolerance target $p$, using the exact threshold changes this ratio to
\[
\frac{q^*(A,p)}{k}.
\]

A reduction of one replica in the replica threshold reduces the encoding traffic by one
replica per object and at the same time reduces the persistent ratio of parity storage by $1/k$.

%
%
%
%
%

\section{Conclusion}\label{sec:conclusion}
This paper studied the recovery threshold of SBIBD based stripeless erasure
coding. We refined the sufficient recovery guarantee by obtaining the exact replica count for $p=1$, identifying direct recovery for $p=2$, and making the two round recovery structure explicit for $p\ge 3$. The resulting guarantee applies uniformly to all legal object
placements.

We then investigated when the universal sufficient replica count can be
reduced. By relating the recovery threshold to structural properties of a
fixed SBIBD representative, we obtained exact thresholds for several
classes. The exact threshold remains open when $p\ge3$, $\lambda>1$, and the fixed
representative~$A$ contains no pairwise separated set of size~$p$. This case
appears to depend on finer incidence information that is not captured by the
present structures, and we leave its exact characterization as an open
problem.

Finally, we studied the realization of the structural conditions used in
the threshold analysis. The existence and relabeling results
connect the theoretical classification with the selection of SBIBDs and
affinity labelings for coding applications.


\bibliography{references}
\bibliographystyle{IEEEtran}

\clearpage

\appendix
\section{Supplementary proofs and examples}
\label{app:supplementary}

The appendix gives supplementary proofs and explicit representatives
referenced in the main text.

\subsection{Existence of zero-diagonal representatives}
\label{app:zero-diagonal}

\begin{proposition}
\label{prop:zero-diagonal-existence}
Every $\SBIBD(v,k,\lambda)$ admits a zero-diagonal representative.
\end{proposition}
\begin{proof}
Let $A$ be an arbitrary representative of the given SBIBD. Let $R$ and $C$ denote the row and column index sets respectively, with
\[
R=C=\{0,1,\cdots,v-1\}.
\]
They are regarded as two distinct parts of a bipartite graph. Define its
edge set by
\[
E:=\{(i,j)\in R\times C:A_{ij}=0\}.
\]

Every row and column of $A$ contains exactly $k$ entries equal to~$1$ and
$v-k$ entries equal to~$0$, thus every vertex has degree
\[
d=v-k>0.
\]
For a subset of row indices $S\subseteq R$, let $N(S) \subseteq  C$ be the set of column indices
adjacent to at least one element of~$S$. 
Every row vertex in $S$ has degree $d$, and all its edges end in
$N(S)$, so there are exactly $d|S|$ edges between $S$ and $N(S)$. Each column vertex in $N(S)$ is incident with at most $d$ of these edges, so
\[
d|S|\le d|N(S)|.
\]
Since $d>0$, we obtain $|S|\le|N(S)|$. Thus Hall's condition holds for every
$S\subseteq R$. Since $|R|=|C|=v$, Hall's theorem gives a perfect matching for the bipartite graph. 
Consequently, there is a bijection
\[
\pi: R \to C
\]
such that
\[
A_{i,\pi(i)}=0 \text{ for every }i\in R.
\]
Relabel old column~$\pi(i)$ as column~$i$, and define
\[
A'_{hi}:=A_{h,\pi(i)}.
\]
Then
\[
A'_{ii}=A_{i,\pi(i)}=0 \text{ for every }i\in R,
\]
so $A'$ is a zero-diagonal representative of the SBIBD.
\end{proof}

\subsection{An obstruction from a reciprocal pair for $p=2$}
\label{app:p2-example}
The complement of the Fano plane has the following zero-diagonal
$\SBIBD(7,4,2)$ representative~\cite{colbourn2010}:
\[
A=\begin{pmatrix}
0 & 1 & 1 & 1 & 0 & 0 & 1 \\
0 & 0 & 1 & 0 & 1 & 1 & 1 \\
0 & 1 & 0 & 1 & 1 & 1 & 0 \\
1 & 1 & 1 & 0 & 1 & 0 & 0 \\
1 & 1 & 0 & 0 & 0 & 1 & 1 \\
1 & 0 & 0 & 1 & 1 & 0 & 1 \\
1 & 0 & 1 & 1 & 0 & 1 & 0
\end{pmatrix}.
\]
Every row and column contains exactly four entries equal to~$1$, and any two
distinct rows contain $1$'s in exactly two common columns. In particular,
$\mathcal{S}_{\mathrm{tgt}}(0)=\{1,2,3,6\}$ and
$\mathcal{S}_{\mathrm{tgt}}(3)=\{0,1,2,4\}$, so
$B_{03}=\{1,2\}$ and nodes $0,3$ form a reciprocal pair. Choose an object
$x$ on node~$0$ with targets $\{1,2,3\}$ and an object $y$ on node~$3$
with targets $\{0,1,2\}$. At nodes~$1$ and~$2$, place $x$ and~$y$ in the
same parity. If nodes~$0$ and~$3$ fail, the parities selected at nodes $3$
and $0$, respectively, are lost, whereas both surviving equations reduce to
$x\XOR y=\gamma$. Their coefficient matrix has rank one. Thus
$q=3=\lambda+1$ fails, and Proposition~\ref{prop:p2-exact} gives
$q^*(A,2)=4$.

\subsection{A representative with a private target permutation on a pairwise separated triple}
\label{app:p3-example}

The following is a zero-diagonal representative of an $\SBIBD(16,6,2)$:
\[
{\setlength{\arraycolsep}{1.5pt}
\renewcommand{\arraystretch}{0.9}
A=\left(\begin{array}{*{16}{c}}
0&1&0&0&0&0&1&0&0&0&1&1&1&0&1&0\\
0&0&0&0&1&0&0&1&1&0&0&1&1&0&0&1\\
0&0&0&1&0&0&0&0&0&1&0&0&1&1&1&1\\
1&0&1&0&0&1&0&0&0&0&0&1&0&0&1&1\\
1&0&0&0&0&0&1&0&1&0&1&0&0&1&0&1\\
1&1&0&0&1&0&0&1&0&0&0&0&0&1&1&0\\
1&1&0&1&0&0&0&0&1&1&0&1&0&0&0&0\\
0&1&1&0&0&1&0&0&1&0&0&0&1&1&0&0\\
0&1&0&1&0&1&0&1&0&0&1&0&0&0&0&1\\
0&0&0&1&1&1&1&0&1&0&0&0&0&0&1&0\\
0&0&0&0&0&1&1&1&0&1&0&1&0&1&0&0\\
1&0&1&1&0&0&1&1&0&0&0&0&1&0&0&0\\
0&1&1&0&1&0&1&0&0&1&0&0&0&0&0&1\\
0&0&1&0&0&0&0&1&1&1&1&0&0&0&1&0\\
0&0&1&1&1&0&0&0&0&0&1&1&0&1&0&0\\
1&0&0&0&1&1&0&0&0&1&1&0&1&0&0&0
\end{array}\right).}
\]
A direct count shows that every row and column has weight~$6$, while any two
distinct rows meet in exactly two columns. For
$F=\{0,1,4\}$, the relevant target sets are
\begin{align*}
\mathcal{S}_{\mathrm{tgt}}(0)&=\{1,6,10,11,12,14\},\\
\mathcal{S}_{\mathrm{tgt}}(1)&=\{4,7,8,11,12,15\},\\
\mathcal{S}_{\mathrm{tgt}}(4)&=\{0,6,8,10,13,15\}.
\end{align*}
Consequently,
\[
B_{01}=\{11,12\},\ B_{04}=\{6,10\},\ B_{14}=\{8,15\}.
\]
These three sets are disjoint, so $F$ is pairwise separated. Moreover,
$0\to1\to4\to0$ is a private target permutation of~$F$, since $A_{0,1}=A_{1,4}=A_{4,0}=1$
and
$d_F(1)=d_F(4)=d_F(0)=1.$

To realize the obstruction at $q=5=\lambda(p-1)+1$, choose objects
$x_0,x_1,x_4$ on nodes $0,1,4$, respectively, with target sets
\begin{align*}
T_{x_0}&=\{1,6,10,11,12\},\\
T_{x_1}&=\{4,8,11,12,15\},\\
T_{x_4}&=\{0,6,8,10,15\}.
\end{align*}
At each target in $B_{ij}$, place $x_i$ and $x_j$ in the same parity, and
fill the remaining parity positions with objects on alive primary nodes.
Now fail the nodes in~$F$. The three cycle parities, hosted at nodes
$1,4,0$, are lost. After the alive encodees have been subtracted, the six
surviving equations have coefficient matrix
\[
M=
\begin{pmatrix}
1&1&0\\
1&1&0\\
1&0&1\\
1&0&1\\
0&1&1\\
0&1&1
\end{pmatrix},\ \text{with columns indexed by }(x_0,x_1,x_4).
\]
Thus $M\cdot (1,1,1)^{\mathsf T}=0$ and $\operatorname{rank}(M)=2<3$, so the
three lost objects are not uniquely determined. Therefore $q=5$ fails,
whereas Theorem~\ref{thm:main} gives recovery at $q=6$. Hence
$q^*(A,3)=6$.

\subsection{A relabeled representative with no private target permutation on any pairwise separated triple}
\label{app:p3-no-private-example}

We next relabel the columns of the same representative in Appendix \ref{app:p3-example} while keeping the
row labels fixed. Define the permutation~$\pi$ by
\[
\begin{split}
\bigl(\pi(0),\pi(1),\cdots,\pi(15)\bigr)
={}&(2,3,11,8,4,5,7,6,9,0,10,1,12,13,15,14),
\end{split}
\]
and let
\[
A'_{it}:=A_{i,\pi(t)},
\text{ for all }i,t\in\{0,1,\cdots,15\}.
\]
Direct check shows
\[
A_{i,\pi(i)}=0, \text{ for every }i\in\{0,1,\cdots,15\},
\]
so $A'$ is also zero-diagonal. Since $A'$ is obtained by permuting the
columns of~$A$, it is another zero-diagonal representative of the same
$\SBIBD(16,6,2)$.

We next verify that no pairwise separated set of size~$3$ admits a private
target permutation under~$A'$. For a node set~$G$, define its incidence degree with respect to~$A'$ by
\[
d'_G(t):=\sum_{i\in G}A'_{it}
\]
and also let
\[
s(G):=
\left|
\left\{
i\in G:
\text{some }t\in G\text{ satisfies }
A'_{it}=1\text{ and }d'_G(t)=1
\right\}
\right|.
\]
A direct enumeration of the $\binom{16}{3}=560$ node triples finds
$240$ pairwise separated triples. Among them, the numbers of triples with
$s(G)=0$, $s(G)=1$, and $s(G)=2$ are respectively $37, 136, 67$ ($37+136+67 = 240$).
In particular, no pairwise separated triple has $s(G)=3$. Since private
targets belonging to different sources are necessarily distinct, a triple
admits a private target permutation exactly when $s(G)=3$. Hence no
pairwise separated triple under~$A'$ admits such a permutation.

Proposition~\ref{prop:private-target-exact} therefore gives
\[
q^*(A',3)=\lambda(3-1)+1=5.
\]
The matrices $A$ and~$A'$ are two zero-diagonal representatives of the same
$\SBIBD(16,6,2)$ obtained from different affinity labelings, but
\[
q^*(A,3)=6,\ q^*(A',3)=5.
\]
Thus the universal replica threshold depends on the affinity labeling as well, not only on the underlying SBIBD.

\vfill

\end{document}